\documentclass[1pt]{article}

\usepackage{amsmath,amssymb}

\newcommand{\eps}{\varepsilon}

\newcommand{\length}{\ensuremath{\mathrm{length}}}
\newtheorem{theorem}{Theorem}
\newtheorem{lemma}{Lemma}

\newenvironment{proof}{\trivlist\item[]\emph{Proof}:}
{\unskip\nobreak\hskip 1em plus 1fil\nobreak$\Box$
\parfillskip=0pt
\endtrivlist}

\begin{document}

\title{Low-Memory Adaptive Prefix Coding}
\author{Travis Gagie\thanks{
Department of Computer Science,  University of Eastern Piedmont.  Email: {\tt travis@mfn.unipmn.it}~.  Supported by Italy-Israel FIRB grant ``Pattern Discovery Algorithms in Discrete Structures, with Applications to Bioinformatics''.}
  \and Marek Karpinski\thanks{
Department of Computer Science, University of Bonn. Email: {\tt \{marek,yasha\}@cs.uni-bonn.de}~.}
\and Yakov
  Nekrich\footnotemark[2]}
\date{}
\maketitle

\begin{abstract}
  In this paper we study the adaptive prefix coding problem in  cases where  
the
   size of the input alphabet is large.  We present an online prefix coding algorithm that uses
  $O(\sigma^{1 / \lambda + \epsilon}) $ bits of space for any constants $\eps>0$, 
  $\lambda>1$, and
  encodes the string of symbols in  $O(\log \log \sigma)$ time per symbol
  \emph{in the worst case}, where $\sigma$ is  the size of the alphabet.
  The upper bound on the encoding length is
  $\lambda n H (s) +(\lambda \ln 2 + 2 + \epsilon) n +
  O (\sigma^{1 / \lambda} \log^2 \sigma)$ bits.
\end{abstract}

\section{Introduction} \label{sec:intro}
In this paper we present an algorithm for  adaptive prefix coding
that uses sublinear space in the size of the alphabet.
Space usage can be an important issue in situations where the
available memory is small; e.g., in mobile computing, when
the alphabet is very large, and when we want the data used by the algorithm to fit into first-level cache memory.

For instance,
Version 5.0 of the Unicode Standard~\cite{Uni06} provides code points
for 99\,089 characters, covering ``all the major languages written
today''.  The Standard itself may be the only document to contain
quite that many distinct characters, but there are over 50\,000
Chinese characters, of which everyday Chinese uses several
thousand~\cite{VZ98}.  One reason there are so many Chinese characters
is that each conveys more information than an English character does; if
we consider syllables, morphemes or words as basic units of text, then
the English `alphabet' is comparably large.  Compressing strings over
such alphabets can be awkward; the problem can be severely aggravated
if we have only a small amount of (cache) memory at our disposal.

Static and adaptive prefix encoding algorithms that
use linear space in the size of the alphabet were extensively studied.
The classical algorithm of Huffman~\cite{Huf52} enables us to construct
an optimal prefix-free code and encode a text
in two passes in $O(n)$ time. Henceforth in this paper,
 $n$ denotes the number of characters in the
 text, and $\sigma$ denotes the size of the alphabet;
 $H(s)= \sum_{i=1}^{\sigma} \frac{f_{a_i}}{n}\log_2\frac{n}{f_{a_i}}$
  is the zeroth-order entropy\footnote{For ease of description, we sometimes
simply denote the entropy by $H$ if the string $s$ is clear from the context.}
 of $s$,  where
$f_a$ denotes the number of occurrences of character $a$ in $s$.
The length of the encoding is $(H+d)n$ bits, and the  redundancy $d$ can be
estimated as $d\leq p_{\max}+0.086$ where  $p_{\max}$ is the
probability of the most frequent character~\cite{G78}.
The drawback to the static Huffman coding is the
need to make two passes over data: we collect the frequencies of
different characters during the first pass, and then construct the
code and encode the string during the second  pass.
Adaptive coding avoids this by maintaining a code for the prefix of the
input string that has already been read and encoded. When a new character $s_i$ is
read, it is encoded with the code for $s_1\ldots s_{i-1}$; then the code
is updated. The FGK algorithm~\cite{Knu85} for adaptive Huffman coding encodes the
string in $(H+2+d)n+O(\sigma\log\sigma)$ bits, while the adaptive Huffman algorithm of
 Vitter~\cite{Vit87}
 guarantees that the string is encoded in $(H+1+d)n+O(\sigma\log\sigma)$ bits.
The adaptive
Shannon coding algorithms  of Gagie~\cite{Gag07} and Karpinski and Nekrich~\cite{KN??}
encode the string in $(H+1)n+O(\sigma\log\sigma)$ bits and
$(H+1)n+O(\sigma\log^2\sigma)$ bits respectively.
All of the above algorithms use space at least linear in the size of the alphabet, to count how often each distinct character occurs.
All algorithms for adaptive prefix coding, with exception of~\cite{KN??},
encode and decode in  $\Theta(nH)$ time, i.e. the time to process the string
depends on $H$ and hence on the size of the input alphabet.
The algorithm of~\cite{KN??} encodes a string in $O(n)$ time, and decoding
takes $O(n\log H)$ time.

Compression with sub-linear space usage was studied by
Gagie and Manzini~\cite{GM07} who proved the following
lower bound:  For any $g$
independent of $n$ and any constants \(\epsilon > 0\) and \(\lambda>1\), 
in the worst
case we cannot encode $s$ in \(\lambda H (s) n + o (n \log \sigma) +
g\) bits if, during the single pass in which we write the
encoding, we use \(O (\sigma^{1 / \lambda - \epsilon})\) bits of
memory.
In~\cite{GM07} the authors also presented an algorithm that divides the
input string into
chunks of length $O(\sigma^{1/\lambda}\log \sigma)$ and
encodes each individual chunk with a modification of the
arithmetic coding, so that the string is encoded with
$(\lambda H(s) +\mu)n + O(\sigma^{1/\lambda}\log \sigma)$ bits.
However, their algorithm is quite complicated and uses arithmetic coding;
hence,  codewords are not self-delimiting and the encoding is not
`instantaneously decodable'. Besides that, their algorithm is
based on static encoding of parts of the input string.

In this paper we present an adaptive prefix coding algorithm that
uses $O(\sigma^{1 / \lambda + \epsilon}) $ bits of memory and encodes a string
$s$ with
$\lambda n H (s) +(\lambda \ln 2 + 2 + \epsilon) n +
O (\sigma^{1 / \lambda} \log^2 \sigma)$ bits.
The encoding and decoding work
in $O(\log \log \sigma)$ time per symbol \emph{in the worst case},
and the whole string $s$ is encoded/decoded in $O(n\log H(s))$ time.
A randomized implementation of our algorithm uses
$O(\sigma^{1 / \lambda} \log^2 \sigma)$ bits of memory and works in
$O(n \log H)$ expected time.
Our method is  based on a simple but effective form of alphabet-partitioning
(see, e.g.,~\cite{CCMW07} and references therein) to trade off the size of a
code and the compression it achieves: we split the alphabet into
frequent and infrequent characters; we preface each occurrence of a
frequent character with a 1, and each occurrence of an infrequent one
with a 0; we replace each occurrence of a frequent character by a
codeword, and replace each occurrence of an infrequent character by
that character's index in the alphabet.

We make a natural assumption
 that unencoded files consist of characters represented by their
indices in  the alphabet (cf. ASCII codes), so we can simply copy the
representation of an infrequent character from the original file.
One difficulty is that we cannot identify the frequent characters
using a low-memory one-pass algorithm: according to the lower
bound of~\cite{KSP03} any online algorithm that identifies
a set of characters $F$, such that each $s\in F$ occurs at
least $\Theta n$   times for some  parameter $\Theta$,
 needs $\Omega(\sigma\log\frac{n}{\sigma})$ bits
of memory in the worst case.  We overcome this difficulty
by maintaining the frequencies of symbols that occur in a
sliding window.

In section~\ref{sec:prelim}, we review the data structures that
are used by our algorithm. In section~\ref{sec:adaptive} we present
a novel encoding method, henceforth called \emph{sliding-window Shannon coding}.
Analysis of the sliding-window Shannon coding is given in
section~\ref{sec:analysis}.

\section{Preliminaries}
\label{sec:prelim}
 The dictionary data structure contains a set $S\subset U$, so
that for any element $x\in U$ we can determine whether $x$ belongs to
$S$. We assume that $|S|=m$. The following dictionary data structure
is described in ~\cite{HMP01}
\begin{lemma} \label{lemma:hmp}
  There exists a $O(m)$ space dictionary data structure
  that can be constructed in $O(m\log m)$ time and supports membership
  queries in $O(1)$ time. \end{lemma}

In the case of a polynomial-size universe, we can easily construct a
data structure that uses more space but also supports updates. The
following Lemma is folklore.
\begin{lemma}\label{lemma:trie}
  If  $|U|=m^{O(1)}$,  then there exists a
  $O(m^{1+\eps})$ space dictionary data structure that can be
  constructed in $O(m^{1+\eps})$ time and supports membership queries
  and updates in $O(1)$ time.
\end{lemma}
\begin{proof} We regard $S$ as a set of binary strings of length $\log
  U$. All strings can be stored in a trie $T$ with node degree
  $2^{\eps'\log U}=m^{\eps}$, where $\eps'=(\log U/\log m)\cdot\eps$.
  The height of $T$ is $O(1)$, and the total number of internal nodes
  is $O(m)$. Each internal node uses $O(m^{\eps})$
space; hence, the
  data structure uses $O(m^{1+\eps})$ space and can be constructed in
  $O(m^{1+\eps})$ time. Clearly, queries and updates are supported in
  $O(1)$ time. \end{proof}

If we allow randomization, then the dynamic $O(m)$ space
dictionary can be maintained. We can use the result of~\cite{DKM+94}:
\begin{lemma}\label{lemma:rand}
There exists a  randomized $O(m)$ space dictionary data structure that supports
 membership queries
  in $O(1)$ time and updates in $O(1)$ expected time.
\end{lemma}
All of the above dictionary data structures can be augmented so that
one or more additional records are associated with each element of $S$;
the record(s) associated with element $a\in S$ can be accessed in $O(1)$ time.

In Section~\ref{sec:adaptive}, we also use the following dynamic partial-sums data structure, due to Moffat~\cite{Mof90}:

\begin{lemma} \label{lem:partial-sums}
There is a dynamic searchable partial-sums data structure that stores a sequence of \(O (\log \sigma)\)-bit real numbers \(p_1, \ldots, p_k\) in \(O (k \log \sigma)\) bits and supports the following operations in \(O (\log i)\) time:
\begin{itemize}
\item given an index $i$, return the $i$-th partial sum \(p_1 + \cdots + p_i\);
\item given a real number $b$, return the index $i$ of the largest partial sum \(p_1 + \cdots + p_i \leq b\);
\item given an index $i$ and a real number $d$, add $d$ to $p_i$.
\end{itemize}
\end{lemma}

\section{Adaptive coding} \label{sec:adaptive}
The adaptive Shannon coding algorithm we present in this section
combines ideas from Karpinski and Nekrich's algorithm~\cite{KN??} with the
sliding-window approach, to encode $s$ in \(\lambda n H (s) +
(\lambda \ln 2 + 2 + \epsilon) n + O (\sigma^{1 / \lambda} \log^2
\sigma)\) bits using \(O (n \log H)\) time overall and \(O (\log \log
\sigma)\) time for any character, \(O (\sigma^{1 / \lambda +
  \epsilon})\) bits of memory and one pass, for any given constants
\(\lambda \geq 1\) and \(\epsilon > 0\).  Whereas Karpinski and
Nekrich's algorithm considers the whole prefix already encoded,
our new algorithm encodes each character \(s [i]\) of $s$
based only on the window \(w_i = s [\max (i - \ell, 1)..(i - 1)]\),
where \(\ell = \left\lceil c \sigma^{1 / \lambda} \log \sigma
\right\rceil\) and $c$ is a constant we will define later in terms of
$\lambda$ and $\epsilon$.  (With \(c = 10\), for example, we produce
an encoding of fewer than \(\lambda n H (s) + (2 \lambda + 2) n + O
(\sigma^{1 / \lambda} \log^2 \sigma)\) bits; with \(c = 100\), the
bound is \(\lambda n H (s) + (0.9 \lambda + 2) n + O (\sigma^{1 /
  \lambda} \log^2 \sigma)\) bits.)  Let \(f (a, s [i..j])\) denote the
number of occurrences of $a$ in \(s [i..j]\).  For \(1 \leq i \leq
n\), if \(f (s [i], w_i) \geq \ell / \sigma^{1 / \lambda}\), then we
write a 1 followed by \(s [i]\)'s codeword in our adaptive Shannon
code; otherwise, we write a 0 followed by \(s [i]\)'s \(\lceil \log
\sigma \rceil\)-bit index in the alphabet.

As in the case of  the quantized Shannon coding~\cite{KN??}, our algorithm maintains
a canonical Shannon code. In a canonical code~\cite{SK64,Con73},
each codeword can be characterized by its length and its position
among codewords of the same length, henceforth called {\em offset}.
The codeword of length $j$ with offset $k$ can be computed
as $\sum_{h=1}^{j-1} n_h/2^h + (k-1)/2^j$.

\tolerance=1000
We maintain four dynamic data structures: a queue $Q$, an augmented
dictionary $D$, an array \(A \left[ \rule{0ex}{2ex} 0..\lceil \log
  \sigma^{1 / \lambda} \rceil, 0..\lfloor \sigma^{1 / \lambda} \rfloor
\right]\) and a searchable partial-sums data structure $P$.  (We
actually use $A$ only while decoding but, to emphasize the symmetry
between the two procedures, we refer to it in our explanation of
encoding as well.)  When we come to encode or decode \(s [i]\),
\begin{itemize}
\item $Q$ stores $w_i$;
\item $D$ stores each character $a$ that occurs in $w_i$, its
  frequency $f (a, w_i)$ there and, if \(f (a, w_i) \geq \ell /
  \sigma^{1 / \lambda}\), its position in $A$;
\item $A[ ]$ is an array of doubly-linked lists. The list $A[j]$,
 \(0 \leq j \leq \lceil \log \sigma^{1 / \lambda} \rceil\),
contains all characters with codeword length $j$ sorted by
the codeword offsets; we denote by $A[j].l$ the pointer to the
last element in $A[j]$.
\item $C[j]$ stores the number of codewords of length $j$
\item $P$ stores \(C [j] / 2^j\) for each $j$ and supports prefix sum
queries.
\end{itemize}
We implement $Q$ in \(O (\ell \log \sigma) = O (\sigma^{1 / \lambda}
\log^2 \sigma)\) bits of memory,  $A$ in \(O (\sigma^{1 /
  \lambda} \log^2 \sigma)\) bits, and $P$ in \(O (\log^2 \sigma)\)
bits by Lemma~\ref{lem:partial-sums}.
The dictionary $D$ uses  \(O (\sigma^{1 / \lambda +
  \epsilon})\) bits and supports queries and updates in $O(1)$
worst-case time
by Lemma~\ref{lemma:trie}; if we allow randomization, we can apply
Lemma~\ref{lemma:rand} and reduce
the space usage to $O(\sigma^{1/\lambda}\log^2\sigma)$ bits,
but updates are supported in $O(1)$ expected time.
Therefore, altogether we use \(O
(\sigma^{1 / \lambda + \epsilon})\) bits of memory;
if randomization is allowed, the space usage is reduced to
$O(\sigma^{1/\lambda}\log^2\sigma)$ bits.

To encode \(s [i]\), we first search in $D$ and, if \(f (s [i], w_i) < \ell / \sigma^{1 / \lambda}\), we simply write a 0
followed by \(s [i]\)'s index in the alphabet, update the data
structures as described below, and proceed to \(s [i + 1]\); if \(f (s [i], w_i) \geq \ell / \sigma^{1 / \lambda}\), we
use $P$ and \(s [i]\)'s position \(A [j, k]\) in $A$ to compute
\[\sum_{h = 0}^{j - 1} C [h] / 2^h + (k - 1) / 2^j \leq 1\,.\]
The first \(j = \lceil \log (\ell / f (s [i], w_i)) \rceil\) bits of
this sum's binary representation are enough to uniquely identify \(s
[i]\) because, if a character \(a \neq s [i]\) is stored at \(A [j',
k']\), then
\[\left| \left(\sum_{h = 0}^{j - 1} C [h] / 2^h + (k - 1) / 2^j \right) -
  \left( \sum_{h = 0}^{j' - 1} C [h ] / 2^h + (k' - 1) / 2^{j'}
  \right) \right| \geq 1 / 2^j\,;\] therefore, we write a 1 followed
by these bits as the codeword for \(s [i]\).

To decode \(s [i]\), we read the next bit in the encoding; if it is a
0, we simply interpret the following \(\lceil \log \sigma \rceil\)
bits as \(s [i]\)'s index in the alphabet, update the data structures,
and proceed to \(s [i + 1]\); if it is a 1, we interpret the following
\(\lceil \log \sigma^{1 / \lambda} \rceil\) bits (of which \(s [i]\)'s
codeword is a prefix) as a binary fraction $b$ and search in $P$ for
index $j$ of the largest partial sum \(\sum_{h = 0}^{j - 1} C [h ] /
2^h \leq b\).  Knowing $j$ tells us the length of \(s [i]\)'s codeword
or, equivalently, its row in $A$; we can also compute its offset,
\[k = \left\lfloor \frac{b - \sum_{h = 0}^{j - 1} C [h] / 2^h}{2^j} \right\rfloor + 1\,;\]
thus, we can find and write \(s [i]\).

Encoding or decoding \(s [i]\) takes \(O (1)\) time for querying $D$
and $A$ and, if \(f (s [i], w_i) \geq \ell / \sigma^{1 / \lambda}\),
then
\[O \left( \rule{0ex}{2ex} \log \log \frac{\ell}{f (s [i], w_i)} \right)
=  O (\log \log \sigma)\] time to query $P$.  After encoding or
decoding \(s [i]\), we update the data structures as follows:
\begin{itemize}
\item
we dequeue \(s [i
- \ell]\) (if it exists) from $Q$ and enqueue \(s [i]\);
we decrement \(s [i - \ell]\)'s frequency in $D$ and delete it if it
does not occur in $w_{i + 1}$; insert \(s [i]\) into $D$ if it does
not occur in $w_i$ or, if it does, increment its frequency;
\item
we remove \(s [i - \ell]\) from $A$ (by replacing it with the last
character in its list \(A[j]\),
decrementing $C[j]$, and
updating $D$) if
\[f (s [i - \ell], w_{i + 1})
< \ell / \sigma^{1 / \lambda} \leq f (s [i - \ell], w_i)\,;\]
\item
we move \(s [i - \ell]\) from list $A[j]$ to list $A[j+1]$ if
\[\lceil \log \sigma^{1 / \lambda} \rceil
\geq \left\lceil \log \frac{\ell}{f (s [i - \ell], w_{i + 1})}
\right\rceil > \left\lceil \log \frac{\ell}{f (s [i - \ell], w_i)}
\right\rceil\,;\]
this is done by replacing $s[i-\ell]$ with $A[j].l$, and appending
$s[i-\ell]$  at the end of $A[j+1]$; pointers $A[j].l$ and $A[j+1].l$
and counters $C[j]$ and $C[j+1]$ are also
updated;
\item
if necessary, we insert \(s [i]\) into $A$ or
move it from $A[j]$ to $A[j+1]$; these procedures
are symmetric to deleting $s[i-\ell]$ and to moving $s[i-\ell]$
from $A[j]$ to $A[j-1]$
\item
finally, if we have changed $C$, the data structure $P$ is updated.
\end{itemize}
  All of these updates, except the last one, take \(O (1)\) time, and
updating $P$ takes \(O (\log \log \sigma)\) time in the worst case.
 When we insert a new
element $s[i]$ into $Q$, this may lead to updating $P$ as described above.
We may decrement the length of $s[i]$ or insert a new codeword for the
symbol $s[i]$. In both cases, we can $P$ updated in $O(length(s[i]))$ time, where
$length(s[i])$ is the current codeword length of $s[i]$.
When we delete an element $s[i-\ell]$,
we may increment the codeword length  of $s[i-\ell]$ or remove it
from the code. If the codeword length is incremented, then we update $P$
in $O(\length(s[i-\ell]))$ time. If we remove the codeword for $s[i-\ell]$, then we also update  $P$ in $O(\length(s[i-\ell]))$ time; in the last case we can
charge the cost of updating $P$ to the previous occurrence of $s[i-\ell]$
in the string $s$, when $s[i-\ell]$ was encoded with $\length(s[i-\ell])$ bits.
The codeword lengths of symbols $s[i]$ and $s[i-\ell]$ are
$O \left( \rule{0ex}{2ex} \log \log \frac{\ell}{f (s [i], w_i)} \right)$
and $O \left( \rule{0ex}{2ex} \log \log \frac{\ell}{f (s [i-\ell], w_i)} \right)$
respectively.
Hence, by Jensen's
inequality, in total we encode $s$ in \(O (n \log H')\) time, where
$H'$ is the average number of bits per character in our encoding.
In the next section, we will prove that the sliding-window Shannon coding
encodes $s$ in \(\lambda n H (s) + (\lambda \ln 2 + 2 + \epsilon) n + O
(\sigma^{1 / \lambda} \log^2 \sigma)\) bits.
Since we can assume that
 $\sigma$ is not vastly larger than $n$,
our method works in $O(n\log H)$ time.

If the dictionary $D$ is implemented as in Lemma~\ref{lemma:rand},
the analysis is exactly the same, but a string $s$ is processed
in expected time $O(n \log H)$.

\begin{lemma}\label{lemma:time}
Sliding-window Shannon coding can be implemented in
 \(O (n \log H)\) time overall and \(O (\log
  \log \sigma)\) time for any character, \(O (\sigma^{1 / \lambda +
    \epsilon})\) bits of memory and one pass.
If randomization is allowed, sliding-window Shannon coding
can be implemented in \(O (\sigma^{1 / \lambda}\log^2\sigma)\)
bits of memory and $O(n\log H)$ expected time.
\end{lemma}

\section{Analysis}
\label{sec:analysis}
In this section we prove the upper bound on the encoding length
of sliding-window Shannon coding and obtain the following Theorem.
\begin{theorem} \label{thm:adaptive}
  We encode $s$ in, and later decode it from, \(\lambda n H (s) +
  (\lambda \ln 2 + 2 + \epsilon) n + O (\sigma^{1 / \lambda} \log^2
  \sigma)\) bits using \(O (n \log H)\) time overall and \(O (\log
  \log \sigma)\) time for any character, \(O (\sigma^{1 / \lambda +
    \epsilon})\) bits of memory and one pass.
  If randomization is allowed, the memory usage can be reduced
  to \(O(\sigma^{1 / \lambda}\log^2\sigma)\)
  bits and $s$ can be encoded and decoded in $O(n\log H)$ expected time.
\end{theorem}

\begin{proof}
  Consider any substring \(s' = s [k..(k + \ell - 1)]\) of $s$ with
  length $\ell$, and let $F$ be the set of characters $a$ such that
\[f \left( \rule{0ex}{2ex} a, s [\max (k - \ell, 1)..(k + \ell - 1)] \right)
\geq \frac{\ell}{\sigma^{1 / \lambda}}\,;\] notice \(|F| \leq 2
\sigma^{1 / \lambda}\).  For \(k \leq i \leq k + \ell - 1\), if \(s
[i] \in F\) but \(f (s [i], w_i) < \ell / \sigma^{1 / \lambda}\), then
we encode \(s [i]\) using
\begin{eqnarray*}
\lefteqn{\lceil \log \sigma \rceil + 1}\\
& < & \lambda \log \sigma^{1 / \lambda} + 2\\
& < & \lambda \log \frac{\ell}
    {\max \left( \rule{0ex}{2ex} f (s [i], w_i), 1 \right)} + 2\\
& \leq & \lambda \log \frac{\ell}
    {\max \left( f \left( \rule{0ex}{2ex} s [i], s [k..(i - 1)] \right), 1 \right)} + 2
\end{eqnarray*}
bits; if \(f (s [i], w_i) \geq \ell / \sigma^{1 / \lambda}\), then we
encode \(s [i]\) using
\[\left\lceil \log \frac{\ell}{f (s [i], w_i)} \right\rceil + 1
< \lambda \log \frac{\ell} {\max \left( f \left( \rule{0ex}{2ex} s
      [i], s [k..(i - 1)] \right), 1 \right)} + 2\] bits; finally, if
\(s [i] \not \in F\), then we again encode \(s [i]\) using
\begin{eqnarray*}
\lefteqn{\lceil \log \sigma \rceil + 1}\\
& < & \lambda \log \sigma^{1 / \lambda} + 2\\
& < & \lambda \log \frac{\ell}
    {f \left( \rule{0ex}{2ex} s [i], s [\max (k - \ell, 1)..(k + \ell - 1)] \right)} + 2\\
& \leq & \lambda \log \frac{\ell}{f (s [i], s')} + 2
\end{eqnarray*}
bits.  Therefore, the total number of bits we use to encode $s'$ is
less than
\begin{eqnarray*}
\lefteqn{\lambda \sum_{a \in F} \sum_{s [i] = a, \atop k \leq i \leq k + \ell - 1} \log \frac{\ell}
    {\max \left( f \left(\rule{0ex}{2ex} a, s [k..(i - 1)] \right), 1 \right)} +}\\
&& \hspace{10ex} \lambda \sum_{a \not \in F} f (a, s') \log \frac{\ell}{f (a, s')} + 2 \ell\\
& = & \lambda \ell \log \ell -
    \lambda \sum_{a \in F} \sum_{s [i] = a, \atop k \leq i \leq k + \ell - 1}
        \log \left( \max \left( f \left( \rule{0ex}{2ex} a, s [k..(i - 1)] \right), 1 \right) \right) -\\
&& \hspace{10ex} \lambda \sum_{a \not \in F} f (a, s_i) \log f (a, s') + 2 \ell\,;
\end{eqnarray*}
since
\[\sum_{s [i] = a, \atop k \leq i \leq k + \ell - 1}
\log \max \left( f \left( \rule{0ex}{2ex} a, s [k..(i - 1)] \right), 1
\right) = \sum_{j = 1}^{f (a, s') - 1} \log j\,,\] we can rewrite our
bound as
\begin{eqnarray*}
\lefteqn{\lambda \left( \ell \log \ell -
    \sum_{a \in F} \sum_{j = 1}^{f (a, s') - 1} \log j -
    \sum_{a \not \in F} f (a, s') \log f (a, s') \right) + 2 \ell}\\
& = & \lambda \left( \ell \log \ell -
    \sum_{a \in F} \log ((f (a, s') - 1)!) -
    \sum_{a \not \in F} f (a, s') \log f (a, s') \right) + 2 \ell\,;
\end{eqnarray*}
by Stirling's Formula,
\begin{eqnarray*}
\lefteqn{\ell \log \ell - \sum_{a \in F} \log ((f (a, s') - 1)!)}\\
& = & \ell \log \ell - \sum_{a \in F} \log ((f (a, s')!) + \sum_{a \in F} \log f (a, s')\\
& \leq & \ell \log \ell -
    \sum_{a \in F} \left( \rule{0ex}{2ex} f (a, s') \log f (a, s') - f (a, s') \ln 2 \right) +
    |F| \log \ell\\
& \leq & \ell \log \ell - \sum_{a \in F} f (a, s') \log f (a, s') +
    \ell \ln 2 + 2 \sigma^{1 / \lambda} \log \ell\,,
\end{eqnarray*}
so we can again rewrite our bound as
\begin{eqnarray*}
\lefteqn{\lambda \left( \ell \log \ell - \sum_a f (a, s') \log f (a, s') +
    \ell \ln 2 + 2 \sigma^{1 / \lambda} \log \ell \right) + 2 \ell}\\
& = & \lambda \sum_a f (a, s') \log \frac{\ell}{f (a, s')} +
    \left( \lambda \ln 2 + 2 + \frac{2 \lambda \sigma^{1 / \lambda} \log \ell}{\ell} \right) \ell\\
& = & \lambda \ell H (s') +
    \left( \lambda \ln 2 + 2 + \frac{2 \lambda \sigma^{1 / \lambda} \log \ell}{\ell} \right) \ell\,.
\end{eqnarray*}
Recall \(\ell = \left\lceil c \sigma^{1 / \lambda} \log \sigma
\right\rceil\), so
\begin{eqnarray*}
\lefteqn{\frac{2 \lambda \sigma^{1 / \lambda} \log \ell}{\ell}}\\
& = & \frac{2 \lambda \sigma^{1 / \lambda}
    \log \left\lceil c \sigma^{1 / \lambda} \log \sigma \right\rceil}
    {\left\lceil c \sigma^{1 / \lambda} \log \sigma \right\rceil}\\
& \leq & \frac{2 \lambda
    \left( \rule{0ex}{2ex}\log c + (1 / \lambda) \log \sigma + \log \log \sigma + 1 \right)}
    {c \log \sigma}\\
& \leq & \frac{2 \lambda (\log c + 3)}{c}
\end{eqnarray*}
(we will give tighter inequalities in the full paper, but use these
here for simplicity); for any constants \(\lambda \geq 1\) and
\(\epsilon > 0\), we can choose a constant $c$ large enough that
\[\frac{2 \lambda (\log c + 3)}{c} < \epsilon\,,\]
so the number of bits we use to encode $s'$ is less than \(\lambda
\ell H (s') + (\lambda \ln 2 + 2 + \epsilon) \ell\). With \(c = 10\),
for example,
\[\frac{2 \lambda (\log c + 3)}{c}
< (2 - \ln 2) \lambda\,,\] so our bound is less than \(\lambda \ell H
(s') + (2 \lambda + 2) \ell\); with \(c = 100\), it is less than
\(\lambda \ell H (s') + (0.9 \lambda + 2) \ell\).

Since the product of length and empirical entropy is superadditive ---
i.e., \(|s_1| H (s_1) + |s_2| H (s_2) \leq |s_1 s_2| H (s_1 s_2)\) ---
we have
\[\ell \sum_{j = 0}^{\lfloor n / \ell \rfloor - 1} H \left( \rule{0ex}{2ex} s [(j \ell + 1)..(j + 1) \ell] \right)
\leq n H (s)\] so, by the bound above, we encode the first \(\ell
\lfloor n / \ell \rfloor\) characters of $s$ using fewer than
\(\lambda n H (s) + (\lambda \ln 2 + 2 + \epsilon) n\) bits.  We
encode the last $\ell$ characters of $s$ using fewer than
\[\lambda \ell H (s [(n - \ell)..n]) + (\lambda \ln 2 + 2 + \epsilon) \ell
= O (\ell \log \sigma) = O (\sigma^{1 / \lambda} \log^2 \sigma)\] bits
so, even counting the bits we use for \(s [(n - \ell + 1)..\ell
\lfloor n / \ell \rfloor]\) twice, in total we encode $s$ using fewer
than
\[\lambda n H (s) + (\lambda \ln 2 + 2 + \epsilon) n + O (\sigma^{1 / \lambda} \log^2 \sigma)\]
bits.
\end{proof}

If the most common $\sigma^{1 / \lambda}$ characters in the alphabet
make up much more than half of $s$ (in particular, when \(\lambda =
1\)) then, instead of using an extra bit for each character, we can
keep a special escape codeword and use it to indicate occurrences of
characters not in the code.  The analysis becomes somewhat
complicated, however, so we leave discussion of this modification for
the full paper.

\section{Summary}
In this paper we presented an algorithm that uses space sub-linear
in the alphabet size and achieves an encoding  length that is close to the
lower bound of~\cite{GM07}.  Our algorithm processes
each symbol  in $O(\log \log \sigma)$ worst-case time,
whereas linear-space prefix coding algorithms can encode
a string of $n$ symbols in $O(n)$ time, i.e. in time independent of the alphabet size $\sigma$.
It is an interesting open problem 
 whether our algorithm (or one with the same space bound) can be made to run in \(O (n)\) time.

\begin{footnotesize}
\bibliographystyle{plain}
\bibliography{lowmem-081108}
\end{footnotesize}

\end{document}